\documentclass[journal]{IEEEtran}

\usepackage{amsmath,amssymb,amsfonts,amsthm,mathtools}
\usepackage{cite}
\usepackage{graphicx}
\usepackage{bm}
\newtheorem{theorem}{Theorem}
\newtheorem{proposition}{Proposition}
\newtheorem{corollary}{Corollary}
\newtheorem{definition}{Definition}
\newtheorem{remark}{Remark}

\newcommand{\R}{\mathbb{R}}
\newcommand{\E}{\mathbb{E}}
\newcommand{\Var}{\mathrm{Var}}
\newcommand{\Cov}{\mathrm{Cov}}
\newcommand{\Qfn}{Q}
\newcommand{\ml}{\mathrm{ML}}
\newcommand{\euc}{\mathrm{Euc}}
\newcommand{\argmin}{\mathop{\mathrm{argmin}}}

\title{Matched and Euclidean-Mismatched Decoding on Fourier-Curve Constellations with Tangent Noise}

\author{
Bin~Han,~\IEEEmembership{Senior Member,~IEEE,} Hao~Chen, Muxia~Sun,~\IEEEmembership{Member,~IEEE,}
\\H.~V.~Poor,~\IEEEmembership{Life Fellow,~IEEE,}
and Hans~D.~Schotten,~\IEEEmembership{Member,~IEEE}%
\thanks{B. Han and H. D. Schotten are with RPTU University Kaiserslautern-Landau, Germany. H. Chen is with ShanghaiTech University, China. M. Sun is with Tsinghua University, China.
H. V. Poor is with Princeton University, USA.
H. D. Schotten is also with the German Research Center for Artificial Intelligence (DFKI), Germany. M. Sun (muxiasun@tsinghua.edu.cn) and H. Chen (chenhao5@shanghaitech.edu.cn) are the corresponding authors. Anthropic's Claude 4.6 Opus and OpenAI's GPT-5.4 assisted with language polishing, auxiliary derivation checking, and minor simulation debugging.}%
}

\begin{document}
\maketitle
\begin{abstract}
We study matched and Euclidean-mismatched decoding on finite Fourier-curve constellations with tangent-space artificial noise. Each hypothesis induces a Gaussian law with symbol-dependent rank-one covariance. We derive exact Euclidean pairwise errors for arbitrary pairs and an exact Gaussian-expectation representation for matched decoding on bilaterally tangent-orthogonal pairs. For uniform even constellations, the Euclidean side yields explicit distance spectra and symbol-error bounds across all offset classes; the matched side is exact on antipodal pairs and benchmarked numerically at the full-codebook level via Monte Carlo. By isolating the detection-theoretic consequence of tangent-space artificial noise, these results clarify analytically how noise fraction and constellation density enter the mismatch behavior; secrecy-rate implications require additional channel and adversary modeling.
\end{abstract}


\section{Introduction}
On a flat or linear constellation, artificial noise (AN) added at the transmitter is statistically identical across hypotheses: when the AN covariance is common across codewords, maximum-likelihood decoding reduces to Euclidean nearest-neighbor decoding whether or not the receiver knows the AN structure. On a curved constellation, this symmetry breaks. When noise is injected along the tangent space of a curved manifold, each hypothesis induces a Gaussian law with a distinct, symbol-dependent rank-one covariance. A receiver that exploits these covariance labels uses a matched metric; a legacy additive white Gaussian noise (AWGN)-designed or complexity-constrained receiver that ignores them falls back to the Euclidean metric. The resulting matched-versus-mismatched gap is a new structural phenomenon created by the interplay of manifold curvature and directional noise.

This work develops an exact analysis of this geometry on the Fourier curve
\begin{equation}
\bm{x}(\theta)=\frac{1}{\sqrt{k}}
\begin{aligned}[t]
\bigl(
&\cos\theta,\sin\theta,\cos2\theta,\sin2\theta,\ldots, \\
&\cos k\theta,\sin k\theta
\bigr),\qquad \theta\in[0,2\pi),
\end{aligned}
\label{eq:fourier-curve}
\end{equation}
with values in $\R^{2k}$, as a canonical tractable instance: its constant speed, explicit harmonic structure, and closed-form antipodal geometry yield exact results in both matched and Euclidean regimes. The structural phenomena---symbol-dependent covariance, tangent--chord alignment, and the mismatch gap---are generic to smooth curved constellations; the Fourier curve makes them analytically accessible. Extensions to other curve families (e.g., parabolic or fractal foldings) would require numerical treatment of the same phenomena.

The tangent-noise mechanism is connected to AN-aided physical-layer security~\cite{goel2008,khisti2010,mukherjee2014}, directional modulation~\cite{daly2009}, constellation design on manifolds~\cite{hochwald2000,love2003}, geometric shaping~\cite{boecherer2015}, mismatched decoding~\cite{merhav1994,scarlett2020,lapidoth1996}, and Gaussian detection with hypothesis-dependent covariance~\cite{poor1994}. This work isolates the detection-theoretic consequence of tangent-space AN; secrecy-rate analysis requires additional modeling of the eavesdropper's channel and lies beyond the present scope.

Specifically, we derive: (i)~exact Euclidean pairwise errors for arbitrary pairs (Proposition~\ref{prop:euc-pep}); (ii)~an exact Gaussian-expectation representation for matched decoding on bilaterally tangent-orthogonal pairs (Theorem~\ref{thm:ml-pep}); (iii)~explicit antipodal geometry for every~$k$ (Theorem~\ref{thm:antipodal-geometry}); and (iv)~explicit Euclidean finite-codebook symbol-error rate (SER) bounds for uniform even constellations (Corollary~\ref{cor:euc-union}). Thus the Euclidean side admits explicit arbitrary-pair and codebook-level bounds, whereas the matched side is exact on phantom/antipodal pairs and benchmarked numerically at the full-codebook level. Because $\det(\bm{\Sigma}_i)$ is hypothesis-invariant, the matched metric differs from Euclidean nearest-neighbor decoding by only a rank-one quadratic correction: one tangent projection and one scalar quadratic term per candidate.

\section{Model and decoders}
Let $k\geqslant 1$, and $\bm{x}(\theta)$ as~\eqref{eq:fourier-curve}. Differentiation gives
\begin{equation}
\bm{x}'(\theta)=\frac{1}{\sqrt{k}}
\bigl(
-\sin\theta,\cos\theta,\ldots,-k\sin k\theta,k\cos k\theta
\bigr).
\label{eq:curve-derivative}
\end{equation}
Its squared speed is independent of $\theta$:
\begin{equation}
\|\bm{x}'(\theta)\|^2=\frac{1}{k}\sum_{m=1}^k m^2=\frac{(k+1)(2k+1)}{6}=:v_k^2.
\label{eq:speed}
\end{equation}
We write
% \begin{equation}
% \hat{\bm{t}}(\theta)=\frac{\bm{x}'(\theta)}{v_k}
% \label{eq:unit-tangent}
% \end{equation}
$\hat{\bm{t}}(\theta)=\bm{x}'(\theta)/v_k$
for the unit tangent.

Let $\Theta=\{\theta_1,\ldots,\theta_M\}\subset[0,2\pi)$ be a finite codebook and define
\[
\bm{x}_i=\bm{x}(\theta_i),
\qquad
\hat{\bm{t}}_i=\hat{\bm{t}}(\theta_i),
\qquad
\bar{\bm{x}}_i=\sqrt{1-\beta}\,\bm{x}_i,
\]
where $\beta\in[0,1)$ is the artificial-noise power fraction. If message $i$ is sent, then
\begin{equation}
\begin{aligned}
\bm{Y}&=\bar{\bm{x}}_i+\sqrt{\beta}\,\xi\,\hat{\bm{t}}_i+\bm{N},\\
\xi&\sim\mathcal{N}(0,1),
\qquad
\bm{N}\sim\mathcal{N}(\bm{0},\sigma_c^2\bm{I}_{2k}),
\end{aligned}
\label{eq:channel}
\end{equation}
with $\xi$ independent of $\bm{N}$. For each harmonic block, $\langle(\cos m\theta,\sin m\theta),(-\sin m\theta,\cos m\theta)\rangle=0$; summing over $m$ gives $\bm{x}_i\perp\bm{x}'_i$ and hence $\bm{x}_i\perp\hat{\bm{t}}_i$. Since $\|\bm{x}_i\|=\|\hat{\bm{t}}_i\|=1$, the average transmit power equals one.

Under hypothesis $i$, the observation law is Gaussian with mean $\bar{\bm{x}}_i$ and covariance
\[
\bm{\Sigma}_i=\beta\,\hat{\bm{t}}_i\hat{\bm{t}}_i^{\top}+\sigma_c^2\bm{I}_{2k}.
\]
Because $\det\bm{\Sigma}_i=\sigma_c^{2(2k-1)}(\beta+\sigma_c^2)$ is independent of $i$, maximum likelihood reduces to minimization of the induced quadratic form. By the Sherman--Morrison identity,
\[
\bm{\Sigma}_i^{-1}=\sigma_c^{-2}\bm{I}_{2k}-\frac{\beta}{\sigma_c^2(\beta+\sigma_c^2)}\hat{\bm{t}}_i\hat{\bm{t}}_i^{\top}.
\]
Therefore the matched receiver uses
\begin{equation}
\widehat{i}_{\ml}
=\argmin_{1\leqslant i\leqslant M}
\left\{
\frac{\|\bm{Y}-\bar{\bm{x}}_i\|^2}{\sigma_c^2}
-\frac{\beta\big((\bm{Y}-\bar{\bm{x}}_i)^{\top}\hat{\bm{t}}_i\big)^2}{\sigma_c^2(\beta+\sigma_c^2)}
\right\}.
\label{eq:ml-metric}
\end{equation}
Relative to Euclidean nearest-neighbor decoding, \eqref{eq:ml-metric} adds only one tangent projection and one scalar quadratic correction per candidate hypothesis. The comparison receiver ignores the covariance anisotropy and uses Euclidean nearest-neighbor decoding:
\begin{equation}
\widehat{i}_{\euc}=\argmin_{1\leqslant i\leqslant M}\|\bm{Y}-\bar{\bm{x}}_i\|^2.
\label{eq:euc-metric}
\end{equation}

Equivalently, \eqref{eq:euc-metric} is the maximum-likelihood rule under any spherically misspecified family $\{\mathcal{N}(\bar{\bm{x}}_i,\tau^2\bm{I}_{2k})\}_{i=1}^M$ with common $\tau^2>0$. Thus \eqref{eq:euc-metric} is the covariance-ignorant AWGN metric obtained by discarding the labels $\{\bm{\Sigma}_i\}$; it also models a legacy AWGN-designed detector, or a complexity-constrained receiver, that ignores symbol-dependent covariance~\cite{merhav1994}. We write
\begin{equation}
\Qfn(x)=\frac{1}{\sqrt{2\pi}}\int_x^{\infty}e^{-u^2/2}\,du.
\label{eq:Q-def}
\end{equation}

\begin{remark}[Implementation]\label{rem:impl}
No online nonlinear manifold transform is required. One may precompute $\{\bar{\bm{x}}_i,\hat{\bm{t}}_i\}$ offline and inject the symbol-dependent rank-one noise in digital baseband; the legitimate receiver then applies~\eqref{eq:ml-metric}, whereas a covariance-ignorant receiver applies~\eqref{eq:euc-metric}. Integration with an outer code is natural but beyond the present work.
\end{remark}

\section{Pairwise formulas}
For two distinct codewords $i\neq j$, define
\begin{equation}
\bm{d}_{ij}=\bm{x}_i-\bm{x}_j,
\qquad
\delta_{ij}=\|\bm{d}_{ij}\|,
\qquad
\cos\alpha_{ij}=\frac{|\bm{d}_{ij}^{\top}\hat{\bm{t}}_i|}{\delta_{ij}}.
\label{eq:dij}
\end{equation}
The angle $\alpha_{ij}$ measures the alignment between the chord and the tangent at the transmitted endpoint.
Throughout this section, $P(i\to j)$ denotes the binary comparison event that hypothesis $j$ scores no worse than hypothesis $i$; at the codebook level these events induce lower and upper bounds by inclusion and union.

\begin{proposition}[Euclidean-mismatched pairwise error]\label{prop:euc-pep}
Conditioned on message $i$ being sent, the Euclidean decoder~\eqref{eq:euc-metric} has pairwise error probability
\begin{equation}
P_{\euc}(i\to j)
=\Qfn\!\left(
\frac{\sqrt{1-\beta}\,\delta_{ij}}{2\sqrt{\beta\cos^2\alpha_{ij}+\sigma_c^2}}
\right).
\label{eq:euc-pep}
\end{equation}
\end{proposition}

\begin{IEEEproof}
Under message $i$, equation~\eqref{eq:channel} becomes
\[
\bm{Y}=\bar{\bm{x}}_i+\sqrt{\beta}\,\xi\,\hat{\bm{t}}_i+\bm{N}.
\]
The event that the Euclidean decoder prefers $j$ to $i$ is
\[
\begin{aligned}
\|\bm{Y}-\bar{\bm{x}}_j\|^2\leqslant \|\bm{Y}-\bar{\bm{x}}_i\|^2
\iff
\sqrt{1-\beta}\,\bm{d}_{ij}^{\top}(\sqrt{\beta}\,\xi\,\hat{\bm{t}}_i+\bm{N}) \\
\leqslant -\frac{(1-\beta)\delta_{ij}^2}{2}.
\end{aligned}
\]
Thus the pairwise error event is
\[
S_{ij}\leqslant -\frac{(1-\beta)\delta_{ij}^2}{2},
\qquad
S_{ij}:=\sqrt{1-\beta}\,\bm{d}_{ij}^{\top}(\sqrt{\beta}\,\xi\,\hat{\bm{t}}_i+\bm{N}).
\]
The variable $S_{ij}$ is centered Gaussian, and independence of $\xi$ and $\bm{N}$ gives
\begin{align*}
\Var(S_{ij})
&=(1-\beta)\Big(\beta(\bm{d}_{ij}^{\top}\hat{\bm{t}}_i)^2+\sigma_c^2\|\bm{d}_{ij}\|^2\Big) \\
&=(1-\beta)\delta_{ij}^2\bigl(\beta\cos^2\alpha_{ij}+\sigma_c^2\bigr).
\end{align*}
The Gaussian tail formula now yields~\eqref{eq:euc-pep}.
\end{IEEEproof}

At $\beta=0$, \eqref{eq:euc-pep} reduces to the classical AWGN pairwise expression $\Qfn(\delta_{ij}/(2\sigma_c))$.

\begin{definition}[Bilaterally tangent-orthogonal pair (phantom pair)]
A pair $(i,j)$ is called a \emph{bilaterally tangent-orthogonal pair}, or a \emph{phantom pair} (because the tangent noise becomes invisible to the Euclidean decoder at such pairs), if
\begin{equation}
\bm{d}_{ij}\perp \hat{\bm{t}}_i
\qquad\text{and}\qquad
\bm{d}_{ij}\perp \hat{\bm{t}}_j.
\label{eq:phantom-condition}
\end{equation}
These are precisely the pairs for which the chord--tangent cross terms vanish in the matched metric difference.
\end{definition}

For a phantom pair, Proposition~\ref{prop:euc-pep} collapses to
\begin{equation}
P_{\euc}(i\to j)=\Qfn\!\left(\frac{\sqrt{1-\beta}\,\delta_{ij}}{2\sigma_c}\right),
\label{eq:euc-phantom}
\end{equation}
so the Euclidean decoder feels artificial noise only through the deterministic signal shrinkage $\sqrt{1-\beta}$.

\begin{theorem}[Matched phantom-pair Gaussian-expectation representation]\label{thm:ml-pep}
Let $(i,j)$ be a phantom pair with Euclidean distance $\delta=\delta_{ij}$ and tangent correlation
\begin{equation}
\gamma=\langle \hat{\bm{t}}_i,\hat{\bm{t}}_j\rangle.
\label{eq:gamma-def}
\end{equation}
Then the matched receiver~\eqref{eq:ml-metric} satisfies the following.

\emph{(i) Degenerate case $|\gamma|=1$.} One has
\begin{equation}
P_{\ml}(i\to j)=\Qfn\!\left(\eta_e\right),
\qquad
\eta_e:=\frac{\sqrt{1-\beta}\,\delta}{2\sigma_c}.
\label{eq:ml-pep-degenerate}
\end{equation}
Hence the matched and Euclidean pairwise tests coincide.

\emph{(ii) Nondegenerate case $|\gamma|<1$.} One has
\begin{equation}
P_{\ml}(i\to j)
=\E\!\left[
\Qfn\!\left(\eta_e+aU^2-bV^2\right)
\right],
\label{eq:ml-pep}
\end{equation}
where $(U,V)$ is a centered jointly Gaussian vector with unit variances and correlation
\begin{equation}
\rho_{UV}=\frac{\gamma\sigma_c}{\sqrt{(1-\gamma^2)\beta+\sigma_c^2}},
\label{eq:rho-uv}
\end{equation}
and
\begin{align}
\eta_e&=\frac{\sqrt{1-\beta}\,\delta}{2\sigma_c},
\label{eq:eta-e}\\
 a&=\frac{\beta\big((1-\gamma^2)\beta+\sigma_c^2\big)}{2\sigma_c(\beta+\sigma_c^2)\delta\sqrt{1-\beta}},
\qquad
 b=\frac{\beta\sigma_c}{2(\beta+\sigma_c^2)\delta\sqrt{1-\beta}}.
\label{eq:ab}
\end{align}
In particular,
\begin{equation}
a-b=\frac{\beta^2(1-\gamma^2)}{2\sigma_c(\beta+\sigma_c^2)\delta\sqrt{1-\beta}}\geqslant 0.
\label{eq:a-minus-b}
\end{equation}
Equality holds when $\beta=0$; for $\beta>0$ and $|\gamma|<1$, one has $a-b>0$.
\end{theorem}

\begin{IEEEproof}
Assume that message $i$ is transmitted, and set
\[
\begin{aligned}
\bm{R}_i&:=\bm{Y}-\bar{\bm{x}}_i=\sqrt{\beta}\,\xi\,\hat{\bm{t}}_i+\bm{N},\\
\bm{R}_j&:=\bm{Y}-\bar{\bm{x}}_j=\sqrt{1-\beta}\,\bm{d}_{ij}+\bm{R}_i.
\end{aligned}
\]
Let $\Lambda_i$ and $\Lambda_j$ denote the matched metrics in~\eqref{eq:ml-metric}. Substituting
\[
\bm{R}_j=\sqrt{1-\beta}\,\bm{d}_{ij}+\bm{R}_i
\]
into the score difference and using the phantom conditions~\eqref{eq:phantom-condition} to eliminate all chord--tangent cross terms gives
\begin{equation}
\begin{aligned}
\Lambda_j-\Lambda_i
&=\frac{(1-\beta)\delta^2+2\sqrt{1-\beta}\,\bm{d}_{ij}^{\top}\bm{N}}{\sigma_c^2}\\
&\quad -\frac{\beta}{\sigma_c^2(\beta+\sigma_c^2)}(C^2-A^2),
\end{aligned}
\label{eq:metric-difference}
\end{equation}
where
\begin{equation}
A:=\bm{R}_i^{\top}\hat{\bm{t}}_i,
\qquad
C:=\bm{R}_j^{\top}\hat{\bm{t}}_j.
\label{eq:A-C-def}
\end{equation}
If $|\gamma|=1$, then $\hat{\bm{t}}_j=\gamma\hat{\bm{t}}_i$. Since the pair is phantom, $\bm{d}_{ij}\perp \hat{\bm{t}}_i$ and therefore also $\bm{d}_{ij}\perp \hat{\bm{t}}_j$. Consequently
\[
C=\bm{R}_j^{\top}\hat{\bm{t}}_j=\gamma\bm{R}_i^{\top}\hat{\bm{t}}_i=\gamma A,
\]
so $A^2-C^2=0$. The score difference~\eqref{eq:metric-difference} reduces to
\[
\Lambda_j-\Lambda_i=\frac{(1-\beta)\delta^2+2\sqrt{1-\beta}\,\bm{d}_{ij}^{\top}\bm{N}}{\sigma_c^2},
\]
and therefore
\[
\{\Lambda_j\leqslant \Lambda_i\}=\{Z\leqslant -\eta_e\}.
\]
This proves~\eqref{eq:ml-pep-degenerate}. We now assume $|\gamma|<1$ and prove~\eqref{eq:ml-pep}.

Hence the error event $\{\Lambda_j\leqslant \Lambda_i\}$ is
\begin{equation}
Z+\frac{\beta(A^2-C^2)}{2\sqrt{1-\beta}\,\sigma_c\delta(\beta+\sigma_c^2)}\leqslant -\eta_e,
\label{eq:error-event-Z}
\end{equation}
where
\begin{equation}
Z:=\frac{\bm{d}_{ij}^{\top}\bm{N}}{\sigma_c\delta}\sim\mathcal{N}(0,1).
\label{eq:Z-def}
\end{equation}

We now factor the quadratic term. Set
\[
s:=\sqrt{1-\gamma^2},
\qquad
W_i:=\hat{\bm{t}}_i^{\top}\bm{N},
\qquad
W_j:=\hat{\bm{t}}_j^{\top}\bm{N}.
\]
Since $\bm{d}_{ij}\perp\hat{\bm{t}}_j$, one has
\[
A=\sqrt{\beta}\,\xi+W_i,
\qquad
C=\sqrt{\beta}\,\gamma\xi+W_j.
\]
Because $|\gamma|<1$, the vector
\[
\hat{\bm{u}}:=\frac{\hat{\bm{t}}_i-\gamma\hat{\bm{t}}_j}{s}
\]
is well defined, has unit norm, and is orthogonal to $\hat{\bm{t}}_j$. Let $W_u:=\hat{\bm{u}}^{\top}\bm{N}$. Since $\bm{N}$ is isotropic, $(W_j,W_u)$ is centered Gaussian with independent $\mathcal{N}(0,\sigma_c^2)$ coordinates, and
\begin{equation}
W_i=\gamma W_j+sW_u.
\label{eq:Wi-decomp}
\end{equation}
Introduce $P:=s\sqrt{\beta}\xi+W_u$, $Q:=\gamma W_u-sW_j$, so that $A=\gamma C+sP$ and $Q=(\gamma A-C)/s$. Equivalently, $P=(A-\gamma C)/s$ and $Q=(\gamma A-C)/s$. 
% \begin{equation}
% P:=s\sqrt{\beta}\,\xi+W_u,
% \qquad
% Q:=\gamma W_u-sW_j.
% \label{eq:PQ-def}
% \end{equation}
% Then
% \begin{equation}
% A=\gamma C+sP,
% \qquad
% Q=\frac{\gamma A-C}{s}.
% \label{eq:A-C-P-Q-rel}
% \end{equation}
% Equivalently,
% \[
% P=\frac{A-\gamma C}{s},
% \qquad
% Q=\frac{\gamma A-C}{s}.
% \]
Therefore
\begin{equation}
P^2-Q^2
=\frac{(A-\gamma C)^2-(\gamma A-C)^2}{1-\gamma^2}
=A^2-C^2.
\label{eq:A2C2}
\end{equation}
Moreover,
\begin{align}
\Var(P)&=(1-\gamma^2)\beta+\sigma_c^2,
\label{eq:varP}\\
\Var(Q)&=\sigma_c^2,
\label{eq:varQ}\\
\Cov(P,Q)&=\gamma\sigma_c^2.
\label{eq:covPQ}
\end{align}
Define the standardized variables
\begin{equation}
U:=\frac{P}{\sqrt{(1-\gamma^2)\beta+\sigma_c^2}},
\qquad
V:=\frac{Q}{\sigma_c}.
\label{eq:UV-def}
\end{equation}
Then $(U,V)$ is centered jointly Gaussian with unit variances and correlation~\eqref{eq:rho-uv}. When $\gamma=0$, the variables $U$ and $V$ are independent.

Finally, $Z$ is independent of $(U,V)$. Indeed,
\[
\hat{\bm{u}}\in\mathrm{span}\{\hat{\bm{t}}_i,\hat{\bm{t}}_j\}\subset \bm{d}_{ij}^{\perp},
\]
so the three unit vectors $\bm{d}_{ij}/\delta$, $\hat{\bm{t}}_j$, and $\hat{\bm{u}}$ are mutually orthogonal. The isotropic Gaussian $\bm{N}$ therefore has independent projections onto these directions, and $\xi$ is independent of $\bm{N}$. Thus $Z$ is independent of $(P,Q)$ and hence of $(U,V)$. Substituting~\eqref{eq:A2C2} and~\eqref{eq:UV-def} into~\eqref{eq:error-event-Z} gives
\[
Z\leqslant -\eta_e-aU^2+bV^2.
\]
Conditioning on $(U,V)$ and using $Z\sim\mathcal{N}(0,1)$ proves~\eqref{eq:ml-pep}. Equation~\eqref{eq:a-minus-b} is immediate from~\eqref{eq:ab}.
\end{IEEEproof}

At $\beta=0$, one has $\eta_e=\delta/(2\sigma_c)$ and $a=b=0$, so \eqref{eq:ml-pep} also collapses to the same AWGN pairwise expression. Equation~\eqref{eq:a-minus-b} is a structural measure of tangent mismatch in the nondegenerate regime $|\gamma|<1$. It quantifies the extra quadratic correction carried by the matched metric, but by itself does not determine the sign or monotonicity of the full expectation~\eqref{eq:ml-pep}. Numerically, one may write $V=\rho_{UV}U+\sqrt{1-\rho_{UV}^2}\,W$ with $U,W\stackrel{\text{i.i.d.}}{\sim}\mathcal{N}(0,1)$ and evaluate the remaining two-dimensional Gaussian integral by Gauss--Hermite quadrature.

\section{Finite-codebook geometry and error-rate analysis}\label{sec:finite}
We now specialize the exact pairwise analysis to finite Fourier-curve constellations.

\begin{figure}[!htbp]
\centering
\includegraphics[width=0.75\linewidth]{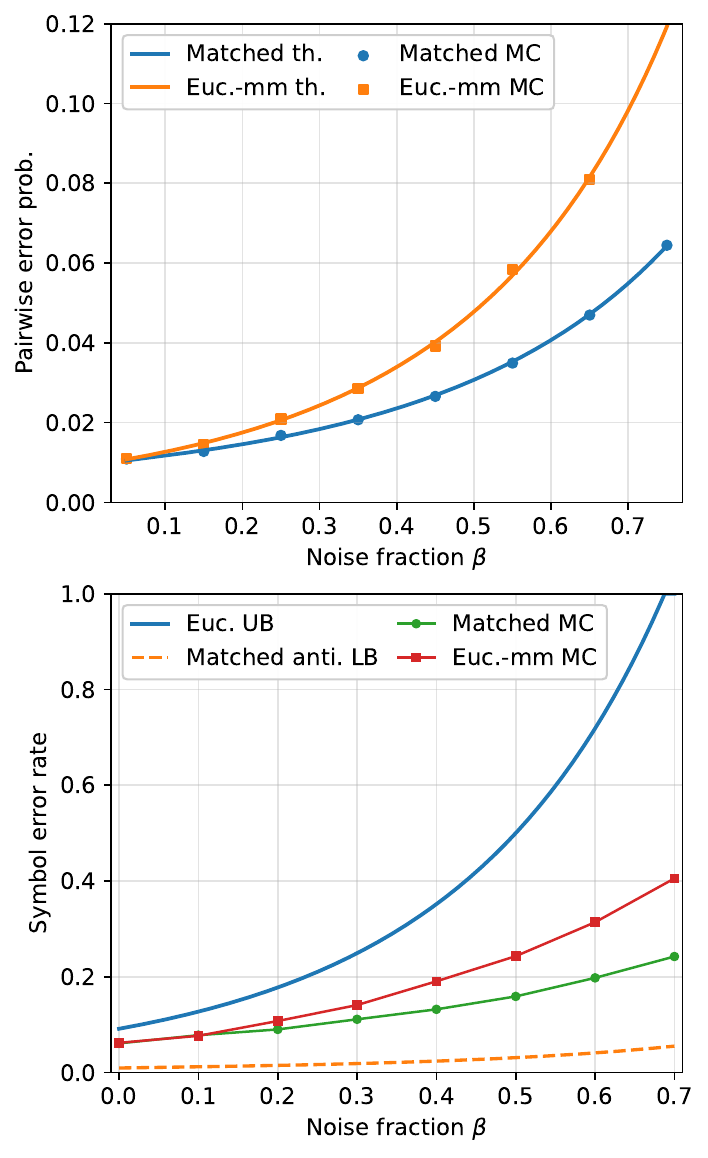}
\caption{Validation of the pairwise and finite-codebook analyses. Above: antipodal pairwise errors for $k=20$ and $\sigma_c=0.3$; lines from Corollary~\ref{cor:antipodal-gap}, markers from Monte Carlo ($5\times 10^4$ trials per $\beta$ value). Below: full-codebook SER for the representative setting $(k,M)=(20,12)$ at $\sigma_c=0.3$; matched and Euclidean-mismatched SER from Monte Carlo simulation ($2\times 10^4$ transmissions per $\beta$ value), the analytic matched antipodal lower bound $P_{\ml}^{\mathrm{anti}}(k)$, and the analytic Euclidean union bound~\eqref{eq:euc-union}.}
\label{fig:phantom-pep}
\end{figure}

\subsection{Antipodal geometry}
Because the speed is constant, the arclength distance between $\theta$ and $\theta+\pi$ along the curve is always $\pi v_k$. The ambient geometry is equally explicit.

\begin{theorem}[Exact antipodal geometry]\label{thm:antipodal-geometry}
For every $k\geqslant 1$ and every $\theta\in[0,2\pi)$, the antipodal difference vector
\begin{equation}
\bm{d}_k(\theta):=\bm{x}(\theta)-\bm{x}(\theta+\pi)
\label{eq:dk-def}
\end{equation}
satisfies
\begin{align}
\langle \bm{d}_k(\theta),\bm{x}'(\theta)\rangle&=0,
&
\langle \bm{d}_k(\theta),\bm{x}'(\theta+\pi)\rangle&=0,
\label{eq:bilateral-orthogonality}\\
\|\bm{d}_k(\theta)\|^2&=\frac{4\lceil k/2\rceil}{k},
\label{eq:delta-k}\\
\langle \hat{\bm{t}}(\theta),\hat{\bm{t}}(\theta+\pi)\rangle&=\frac{3(-1)^k}{2k+1}.
\label{eq:gamma-k}
\end{align}
Consequently the antipodal distortion ratio is
\begin{equation}
\rho_k:=\frac{\pi v_k}{\|\bm{d}_k(\theta)\|}
=\frac{\pi}{2}\sqrt{\frac{k(k+1)(2k+1)}{6\lceil k/2\rceil}}
=\frac{\pi k}{\sqrt{6}}+O(1).
\label{eq:rho-k}
\end{equation}
\end{theorem}

\begin{IEEEproof}
For the $m$th harmonic block, $\bm{x}_m(\theta+\pi)=(-1)^m\bm{x}_m(\theta)$, so $\bm{d}_{k,m}(\theta)=\frac{1-(-1)^m}{\sqrt{k}}(\cos m\theta,\sin m\theta)$. Each block satisfies $\langle\bm{d}_{k,m},\bm{x}_m'\rangle=0$ by direct computation; summing over $m$ proves~\eqref{eq:bilateral-orthogonality}. Only odd harmonics contribute to $\|\bm{d}_k\|^2$, giving~\eqref{eq:delta-k}. The alternating sum $\sum_{m=1}^k(-1)^m m^2=(-1)^k k(k+1)/2$, obtained by pairing consecutive terms, yields~\eqref{eq:gamma-k} after normalization by $v_k^2$. Formula~\eqref{eq:rho-k} follows from the arclength $\pi v_k$ and the chord~\eqref{eq:delta-k}.
\end{IEEEproof}

\begin{corollary}[Antipodal pairwise formulas]\label{cor:antipodal-gap}
Fix $k\geqslant 1$, and let $j$ correspond to the antipodal point $\theta_j=\theta_i+\pi\pmod{2\pi}$. Then $(i,j)$ is a phantom pair. Its Euclidean distance and tangent correlation are
\begin{equation}
\delta_k=2\sqrt{\frac{\lceil k/2\rceil}{k}},
\qquad
\gamma_k=\frac{3(-1)^k}{2k+1}.
\label{eq:delta-gamma-k}
\end{equation}
Hence
\begin{equation}
P_{\euc}^{\mathrm{anti}}(k)
=\Qfn\!\left(
\frac{\sqrt{1-\beta}}{\sigma_c}\sqrt{\frac{\lceil k/2\rceil}{k}}
\right).
\label{eq:euc-antipodal-final}
\end{equation}
If $k=1$, then $\gamma_1=-1$ and
\begin{equation}
P_{\ml}^{\mathrm{anti}}(1)=P_{\euc}^{\mathrm{anti}}(1)=\Qfn\!\left(\frac{\sqrt{1-\beta}}{\sigma_c}\right).
\label{eq:ml-antipodal-k1}
\end{equation}
If $k\geqslant 2$, then
\begin{equation}
P_{\ml}^{\mathrm{anti}}(k)
=\E\!\left[\Qfn\!\left(\eta_{e,k}+a_kU^2-b_kV^2\right)\right],
\label{eq:ml-antipodal-final}
\end{equation}
where $\eta_{e,k},a_k,b_k$ are obtained from~\eqref{eq:eta-e}--\eqref{eq:ab} by substituting $(\delta,\gamma)=(\delta_k,\gamma_k)$. Moreover, for every $k\geqslant 2$,
\begin{equation}
\begin{aligned}
a_k-b_k
&=\frac{\beta^2(1-\gamma_k^2)}{2\sigma_c(\beta+\sigma_c^2)\delta_k\sqrt{1-\beta}}\geqslant 0,\\
|\gamma_k|&=\frac{3}{2k+1}\downarrow 0.
\end{aligned}
\label{eq:ak-bk}
\end{equation}
Equality in the first relation holds when $\beta=0$; for $\beta>0$ and $k\geqslant 2$, one has $a_k-b_k>0$. Thus the explicit mismatch parameter approaches the orthogonal-tangent regime $\gamma=0$ as $k\to\infty$.
\end{corollary}

\begin{IEEEproof}
The phantom property follows from~\eqref{eq:bilateral-orthogonality}. The formulas for $\delta_k$ and $\gamma_k$ come from~\eqref{eq:delta-k} and~\eqref{eq:gamma-k}. Substituting $\delta_k$ into~\eqref{eq:euc-phantom} gives~\eqref{eq:euc-antipodal-final}. If $k=1$, then $\gamma_1=-1$, so~\eqref{eq:ml-antipodal-k1} follows from the degenerate case~\eqref{eq:ml-pep-degenerate}. If $k\geqslant 2$, then $|\gamma_k|<1$, so substituting $(\delta,\gamma)=(\delta_k,\gamma_k)$ into~\eqref{eq:ml-pep} gives~\eqref{eq:ml-antipodal-final}. Equation~\eqref{eq:ak-bk} is the specialization of~\eqref{eq:a-minus-b}.
\end{IEEEproof}

As $k\to\infty$, antipodal pairs approach a universal orthogonal-tangent limit: $\delta_k\to\sqrt{2}$ and $\gamma_k\to 0$. Thus increasing $k$ does not enlarge the antipodal chord; it drives the pairwise geometry toward the orthogonal-tangent mismatch regime.

\subsection{Uniform even constellations and explicit SER analysis}
Let
\begin{equation}
\Theta_M:=\left\{\theta_m=\frac{2\pi m}{M}:m=0,1,\ldots,M-1\right\},
\qquad M\ \text{even}.
\label{eq:uniform-codebook}
\end{equation}
Each symbol then has a unique antipodal partner with index offset $M/2$.

For an offset $q\in\{1,\ldots,M-1\}$ define
\begin{equation}
\Delta_q:=\frac{2\pi q}{M},
\qquad
\delta_{k,M}(q):=\|\bm{x}(\theta)-\bm{x}(\theta+\Delta_q)\|,
\label{eq:Delta-q}
\end{equation}
where the right-hand side is independent of $\theta$ by rotational symmetry.

\begin{proposition}[Uniform offset geometry]\label{prop:uniform-offset}
For the uniform codebook~\eqref{eq:uniform-codebook} and any $q\in\{1,\ldots,M-1\}$,
\begin{align}
\delta_{k,M}(q)^2
&=\frac{2}{k}\sum_{m=1}^k\bigl(1-\cos(m\Delta_q)\bigr)
\label{eq:delta-spectrum-a}\\
&=2-\frac{2}{k}
\frac{\sin(k\Delta_q/2)\cos((k+1)\Delta_q/2)}{\sin(\Delta_q/2)},
\label{eq:delta-spectrum-b}
\end{align}
and
\begin{equation}
\cos\alpha_{k,M}(q)
=\frac{1}{k v_k\,\delta_{k,M}(q)}
\left|\sum_{m=1}^k m\sin(m\Delta_q)\right|
\label{eq:alpha-spectrum-a}
\end{equation}
with the closed form
\begin{equation}
\cos\alpha_{k,M}(q)
=\frac{\left|(k+1)\sin(k\Delta_q)-k\sin((k+1)\Delta_q)\right|}{4k v_k\,\delta_{k,M}(q)\sin^2(\Delta_q/2)}.
\label{eq:alpha-spectrum-b}
\end{equation}
Consequently,
\begin{equation}
P_{\euc}^{(q)}(k,M)
=\Qfn\!\left(
\frac{\sqrt{1-\beta}\,\delta_{k,M}(q)}{2\sqrt{\beta\cos^2\alpha_{k,M}(q)+\sigma_c^2}}
\right)
\label{eq:euc-q}
\end{equation}
is explicit for every offset $q$.
\end{proposition}

\begin{IEEEproof}
For the $m$th harmonic block,
\[
\|\bm{x}_m(\theta)-\bm{x}_m(\theta+\Delta_q)\|^2
=\frac{2}{k}\bigl(1-\cos(m\Delta_q)\bigr).
\]
Summing over $m$ gives~\eqref{eq:delta-spectrum-a}, and the finite cosine-sum identity yields~\eqref{eq:delta-spectrum-b}. Likewise,
\[
\bigl(\bm{x}(\theta)-\bm{x}(\theta+\Delta_q)\bigr)^{\top}\bm{x}'(\theta)
=-\frac{1}{k}\sum_{m=1}^k m\sin(m\Delta_q),
\]
which implies~\eqref{eq:alpha-spectrum-a} after division by $v_k\delta_{k,M}(q)$. The finite sine-sum identity
\[
\sum_{m=1}^k m\sin(m\Delta_q)
=\frac{(k+1)\sin(k\Delta_q)-k\sin((k+1)\Delta_q)}{4\sin^2(\Delta_q/2)}
\]
gives~\eqref{eq:alpha-spectrum-b}. Substituting~\eqref{eq:delta-spectrum-b} and~\eqref{eq:alpha-spectrum-b} into Proposition~\ref{prop:euc-pep} yields~\eqref{eq:euc-q}.
\end{IEEEproof}

\begin{corollary}[Explicit Euclidean SER bounds]\label{cor:euc-union}
Consider the uniform even codebook~\eqref{eq:uniform-codebook} with equiprobable symbols. Let $P_{s,\euc}(k,M)$ be the symbol-error rate of the Euclidean decoder. Then
\begin{equation}
\begin{aligned}
\max_{1\leqslant q\leqslant M/2} P_{\euc}^{(q)}(k,M)
&\leqslant P_{s,\euc}(k,M)\\
&\leqslant 2\sum_{q=1}^{M/2-1} P_{\euc}^{(q)}(k,M)+P_{\euc}^{\mathrm{anti}}(k).
\end{aligned}
\label{eq:euc-union}
\end{equation}
where $P_{\euc}^{(q)}(k,M)$ is given by~\eqref{eq:euc-q}. Moreover,
$P_{\euc}^{\mathrm{anti}}(k)=P_{\euc}^{(M/2)}(k,M)$ is the explicit antipodal term~\eqref{eq:euc-antipodal-final}.
\end{corollary}

\begin{IEEEproof}
By rotational symmetry of the uniform codebook, the pairwise probability from any symbol to the symbol at offset $q$ depends only on $q$. For each $q$, the event that the offset-$q$ symbol beats the transmitted symbol implies symbol error; hence
\[
P_{\euc}^{(q)}(k,M)\leqslant P_{s,\euc}(k,M),\qquad q=1,\ldots,M/2.
\]
Taking the maximum over $q$ gives the lower bound in~\eqref{eq:euc-union}. The union bound gives
\[
P_{s,\euc}(k,M)\leqslant \sum_{q=1}^{M-1}P_{\euc}^{(q)}(k,M).
\]
Since $P_{\euc}^{(q)}(k,M)=P_{\euc}^{(M-q)}(k,M)$, pairing offsets $q$ and $M-q$ yields the upper bound in~\eqref{eq:euc-union}. The offset $q=M/2$ is exactly the antipodal term.
\end{IEEEproof}

The Euclidean side is analytically explicit across all offset classes: Corollary~\ref{cor:euc-union} gives explicit lower and upper Euclidean SER bounds on every uniform even constellation. For the matched decoder, the antipodal event yields the lower bound $P_{\ml}^{\mathrm{anti}}(k)\leqslant P_{s,\ml}(k,M)$. In a uniform even codebook, this explicit matched formula covers only the antipodal offset class $q=M/2$, i.e., one offset class out of the $M-1$ pairwise comparisons seen by each symbol; a corresponding matched full-codebook expression for general uniform offsets remains unavailable.

Fig.~\ref{fig:phantom-pep} provides the numerical codebook-level comparison. Its left panel validates the antipodal pairwise formulas, while its right panel compares full-codebook matched and Euclidean-mismatched SER on the representative uniform even constellation $(k,M)=(20,12)$. The Euclidean union bound becomes looser at larger $\beta$ because pairwise error events overlap more strongly, but it still tracks the $\beta$-dependence of the Euclidean SER.

For fixed $k$ and small offset $\Delta_q$, Taylor expansion of~\eqref{eq:delta-spectrum-a} gives
\begin{equation}
\delta_{k,M}(q)=v_k\Delta_q-\frac{v_k(3k^2+3k-1)}{120}\,\Delta_q^3+O(\Delta_q^5).
\label{eq:local-spacing}
\end{equation}
Since $v_k=\sqrt{(2k^2+3k+1)/6}=k/\sqrt{3}+\sqrt{3}/4+O(k^{-1})$, for fixed $\sigma_c$ and fixed small offset class $q$ (especially $q=1$) one has $\Delta_q=2\pi q/M$ and hence $\delta_{k,M}(1)\approx (2\pi/\sqrt{3})(k/M)$ for moderate and large $k$. Thus keeping $M/k$ approximately constant stabilizes the nearest-neighbor spacing at first order, while the actual error performance still depends materially on both $\beta$ and $\sigma_c$.

\section{Conclusion}
When artificial noise is injected along the tangent space of a curved constellation, each hypothesis acquires a distinct rank-one covariance---a structural asymmetry absent on flat codebooks. This work has developed an exact pairwise analysis together with explicit Euclidean finite-codebook bounds for matched-versus-mismatched decoding on Fourier-curve constellations: exact Euclidean pairwise errors for arbitrary pairs, an exact Gaussian-expectation matched formula on tangent-orthogonal pairs, explicit antipodal geometry for every~$k$, and explicit Euclidean SER bounds for uniform even constellations.

Two structural findings emerge. First, increasing $k$ drives antipodal pairs toward an orthogonal-tangent mismatch regime ($|\gamma_k|=3/(2k+1)\to 0$, $\delta_k\to\sqrt{2}$) rather than enlarging the chord. Second, the matched metric differs from Euclidean nearest-neighbor decoding by only one tangent projection and one scalar quadratic correction per candidate, yet this low-cost correction yields a measurable SER improvement. These results directly support model-aware decoding analysis on curved constellations; secrecy-rate and adversarial channel extensions require additional modeling and are left for future work.


\begingroup\scriptsize
\begin{thebibliography}{99}

\bibitem{goel2008}
S.~Goel and R.~Negi, ``Guaranteeing secrecy using artificial noise,'' \emph{IEEE Trans. Wireless Commun.}, vol.~7, no.~6, pp.~2180--2189, 2008.

\bibitem{khisti2010}
A.~Khisti and G.~W. Wornell, ``Secure transmission with multiple antennas---II: The MIMOME wiretap channel,'' \emph{IEEE Trans. Inf. Theory}, vol.~56, no.~11, pp.~5515--5532, 2010.

\bibitem{mukherjee2014}
A.~Mukherjee, S.~A.~A. Fakoorian, J.~Huang, and A.~L. Swindlehurst, ``Principles of physical layer security in multiuser wireless networks: A survey,'' \emph{IEEE Commun. Surveys Tuts.}, vol.~16, no.~3, pp.~1550--1573, 2014.

\bibitem{daly2009}
M.~P. Daly and J.~T. Bernhard, ``Directional modulation technique for phased arrays,'' \emph{IEEE Trans. Antennas Propag.}, vol.~57, no.~9, pp.~2633--2640, 2009.

\bibitem{hochwald2000}
B.~M. Hochwald and T.~L. Marzetta, ``Unitary space-time modulation for multiple-antenna communications in Rayleigh flat fading,'' \emph{IEEE Trans. Inf. Theory}, vol.~46, no.~2, pp.~543--564, 2000.

\bibitem{love2003}
D.~J. Love, R.~W. Heath Jr., W.~Santipach, and M.~L. Honig, ``Grassmannian beamforming for multiple-input multiple-output wireless systems,'' \emph{IEEE Trans. Inf. Theory}, vol.~49, no.~10, pp.~2735--2747, 2003.

\bibitem{boecherer2015}
G.~B\"ocherer, F.~Steiner, and P.~Schulte, ``Bandwidth efficient and rate-matched low-density parity-check coded modulation,'' \emph{IEEE Trans. Commun.}, vol.~63, no.~12, pp.~4651--4665, 2015.

\bibitem{merhav1994}
N.~Merhav, G.~Kaplan, A.~Lapidoth, and S.~Shamai, ``On information rates for mismatched decoders,'' \emph{IEEE Trans. Inf. Theory}, vol.~40, no.~6, pp.~1953--1967, 1994.

\bibitem{scarlett2020}
J.~Scarlett, A.~Martinez, and A.~Guill\'en i F\`abregas, ``Mismatched decoding: Error exponents, second-order rates and saddlepoint approximations,'' \emph{Found. Trends Commun. Inf. Theory}, vol.~17, nos.~2--3, pp.~93--300, 2020.

\bibitem{lapidoth1996}
A.~Lapidoth, ``Nearest neighbor decoding for additive non-Gaussian noise channels,'' \emph{IEEE Trans. Inf. Theory}, vol.~42, no.~5, pp.~1520--1529, 1996.

\bibitem{poor1994}
H.~V. Poor, \emph{An Introduction to Signal Detection and Estimation}, 2nd~ed. New York, NY, USA: Springer, 1994.

\end{thebibliography}
\endgroup
\end{document}